\documentclass[11pt]{article}

\usepackage{graphicx}
\usepackage{epstopdf}
\usepackage{amssymb}
\usepackage{amsthm}
\usepackage{authblk}
\usepackage[ruled,linesnumbered,algo2e,vlined]{algorithm2e}

\newtheorem{theorem}{Theorem}
\newtheorem{lemma}{Lemma}

\newcommand{\etal}{et~al.}

\newcommand{\weight}[1]{\mathsf{weight}(#1)}

\newcommand{\fnum}{u_f}
\newcommand{\onum}{u_o}
\newcommand{\AC}{\mathit{\pi}}
\newcommand{\ACautomaton}[1]{\mathsf{AC}(#1)}

\newcommand{\failure}{\mathsf{flink}}
\newcommand{\out}{\mathsf{output}}
\newcommand{\rootState}{\mathit{rootState}}
\newcommand{\activeState}{\mathit{activeState}}
\newcommand{\newState}{\mathit{newState}}

\newcommand{\newStatesSet}{\mathit{newStatesSet}}
\newcommand{\deleteStatesSet}{\mathit{deleteStatesSet}}

\newcommand{\outStates}{\mathit{outStates}}
\newcommand{\failStates}{\mathit{failStates}}

\newcommand{\getOutStates}{\mathsf{getOutStates}}
\newcommand{\getFailStates}{\mathsf{getFailStates}}
\newcommand{\outNode}{\mathit{outNode}}
\newcommand{\failureState}{\mathit{failureState}}

\newcommand{\state}{s}

\newcommand{\suf}{\mathsf{slink}}

\newcommand{\rootNode}{\mathit{root}}
\newcommand{\activeNode}{\mathit{activeNode}}

\newcommand{\bnum}{\mathit{g_n}}
\newcommand{\tnum}{\mathit{g_t}}

\newcommand{\node}{\mathit{node}}
\newcommand{\goto}[2]{\mathsf{goto}(#1,#2)}
\newcommand{\stack}{\mathit{stack}}
\newcommand{\queue}{\mathit{queue}}
\newcommand{\trans}{\mathsf{trans}}
\newcommand{\start}{\mathit{start}}
\newcommand{\slink}{\mathit{slink}}
\newcommand{\lnode}{\mathit{lnode}}
\newcommand{\fail}{\mathbf{fail}}

\newcommand{\chk}[2]{\mathsf{checkTransition(#1,#2)}}

\newcommand{\Substr}[1]{\mathsf{Substr}(#1)}
\newcommand{\PrefixSet}[1]{\mathsf{Pref}(#1)}
\newcommand{\ePos}[2]{\mathsf{endPos}_{#2}(#1)}
\newcommand{\DAWG}[1]{\mathsf{DAWG}(#1)}
\newcommand{\isuf}{\mathsf{isuf}}

\newcommand{\nma}{\mathit{NMA}}
\newcommand{\Null}{\mathsf{NULL}}
\newcommand{\Output}{\mathsf{output}}
\newcommand{\True}{\mathbf{true}}
\newcommand{\False}{\mathbf{false}}
\newcommand{\EmptySet}{\emptyset}

\newcommand{\occ}{\mathit{occ}}

\newcommand{\dadd}{\mathit{d_{add}}}
\newcommand{\utotal}{\mathit{total\_u_f}}

	\usepackage[top=3cm, bottom=2cm, left=3.5cm, right=3.5cm, includefoot]{geometry}
	\usepackage{hyperref}
	\usepackage[
	type={CC},
	modifier={by-nc-nd},
	version={4.0},
	]{doclicense}
	
	\makeatletter
	\def\blfootnote{\xdef\@thefnmark{}\@footnotetext}
	\makeatother
	
\begin{document}
	\title{Efficient Dynamic Dictionary Matching \\ with DAWGs and AC-automata}

	\author[1]{Diptarama Hendrian}
	\author[2]{Shunsuke Inenaga}
	\author[1]{Ryo Yoshinaka}
	\author[1]{\\Ayumi Shinohara}
	\affil[1]{Graduate School of Information Sciences, Tohoku University\\
		 Sendai, Japan\\
	 	\texttt{\{diptarama@shino., ry@, ayumi@\}ecei.tohoku.ac.jp}}
	\affil[2]{Department of Informatics, Kyushu University\\
		Fukuoka, Japan\\
		\texttt{inenaga@inf.kyushu-u.ac.jp}}

	\date{}
	\maketitle
	
	\begin{abstract}

The \emph{dictionary matching} is a task
to find all occurrences of pattern strings in a set $D$ (called a dictionary) on a text string $T$.
The Aho-Corasick-automaton (AC-automaton)
which is built on $D$
is a fundamental data structure which enables us to solve the dictionary matching problem in $O(d\log\sigma)$ preprocessing time and $O(n\log\sigma  + \occ)$ matching time,
where $d$ is the total length of the patterns in the dictionary $D$,
$n$ is the length of the text,
$\sigma$ is the alphabet size,
and $\occ$ is the total number of occurrences of all the patterns in the text.
The \emph{dynamic dictionary matching} is a variant where patterns may dynamically be inserted into and deleted from the dictionary $D$.
This problem is called \emph{semi-dynamic dictionary matching} if only insertions are allowed.
In this paper, we propose two efficient algorithms that can solve both problems with some modifications.
For a pattern of length $m$,
our first algorithm supports insertions in
$O(m\log \sigma + \log d / \log \log d)$ time and pattern matching in $O(n \log \sigma + \occ)$
for the semi-dynamic setting.
This algorithm also supports both insertions and deletions in
$O(\sigma m + \log d / \log \log d)$ time and 
pattern matching in $O(n (\log d / \log \log d + \log \sigma) + \occ (\log d / \log \log d))$ time for the dynamic dictionary matching problem by some modifications.
This algorithm is based on the \emph{directed acyclic word graph} (\emph{DAWG})
of Blumer et al. (JACM 1987).
Our second algorithm, which is based on the AC-automaton,
supports insertions in $O(m\log \sigma + \fnum + \onum)$ time for the semi-dynamic setting
and supports both insertions and deletions in $O(\sigma m + \fnum + \onum)$ time for the dynamic setting,
where $\fnum$ and $\onum$ respectively denote the numbers of
states in which the failure function and the output function
need to be updated.
This algorithm performs pattern matching in $O(n \log \sigma + \occ)$ time for both settings.
Our algorithm achieves optimal update time for AC-automaton based methods
over constant-size alphabets,
since \emph{any} algorithm
which explicitly maintains the AC-automaton requires 
$\Omega(m + \fnum + \onum)$ update time.
 
	\end{abstract}
	\blfootnote{\textcopyright 2018. This manuscript version is made available under the \doclicenseNameRef license
		
	\doclicenseImage[imagewidth=3cm]}
\section{Introduction} \label{sec:intro}

The pattern matching problem is, 
given a pattern string and a text string,
to output all occurrence positions of the pattern in the text.
Pattern matching is one of the most fundamental problems in string processing
and has been extensively studied for decades.
Efficient pattern matching algorithms are primitives for
various applications such as data mining, search engines, text editors, etc.

A natural extension of the pattern matching problem is to
consider a set of multiple patterns.
That is, given a set $D=\{p_1,p_2,\ldots,p_r\}$ of patterns called
a \emph{dictionary} and a single text,
the task is to find all occurrence positions of each pattern $p_i \in D$
in the text.
This problem is called the \emph{dictionary matching problem}~\cite{aho1975efficient,amir1994dynamic}.
Aho and Corasick's algorithm (AC-algorithm)~\cite{aho1975efficient}
and Commentz-Walter's algorithm~\cite{Commentz-Walter1979}
are fundamental solutions to the dictionary matching problem.
For a dictionary of size $d$ over an alphabet of size $\sigma$,
both of the above algorithms first preprocess the dictionary in $O(d\log\sigma)$ time.
Then, given a text of length $n$,
the occurrences of all patterns $p_i \in D$ in the text
can be reported in $O(n\log\sigma + \occ)$ time by the AC-algorithm,
and in $O(nd\log\sigma)$ time by the Commentz-Walter algorithm,
where $\occ$ is the total number of occurrences of all the patterns in the text.
Notice that $\occ \leq nd$ always holds and hence the $\occ$ term is omitted
in the latter time complexity.

Meyer~\cite{meyer1985incremental} introduced the incremental string matching problem, which is also known as the \emph{semi-dynamic dictionary matching problem},
a variant of dictionary matching that allows insertion of a pattern into the dictionary.
He proposed an algorithm for semi-dynamic dictionary matching,
which updates the AC-automaton when a new pattern is inserted
into the dictionary.
Amir \etal~\cite{amir1994dynamic} introduced the \emph{dynamic dictionary matching problem} which allows for both insertion and deletion of patterns.
Several sophisticated data structures for dynamic dictionary matching
have been proposed in the literature~\cite{amir1994dynamic,amir1995improved,chan2005dynamic,idury1994dynamic,AlstrupHR98,AlstrupHR98tr}.
All these data structures use linear $O(d)$ (words of) space.
More recently, succinct data structures for dynamic dictionary matching have been produced, where the main concern is to store the dictionary in memory space close to the information theoretical minimum~\cite{feigenblat2014improved,hon2009succinct}.

\begin{table}[tb]
	\label{table:summary}
	\begin{center}
		\caption{Comparison of the algorithms for the dynamic dictionary matching. The last four in the table are our proposed methods.
		  Here, $n$ is the length of the text, $d$ is the total length of the dictionary of patterns, and $m$ is the length of the pattern to insert or delete. $k$ and $\epsilon$ are any constants with $k \geq 2$ and $0 < \epsilon < 1$, respectively. $l_{\max}$ is the length of the longest pattern in the dictionary, and $z$ is the size of the AC-automaton before updates.
			$\dagger$ indicates algorithms for semi-dynamic dictionary matching which allows only for insertion. The bounds for Chan \etal's algorithm~\cite{ChanHLS07} hold for constant alphabets.}
		\begin{tabular}{|l|c|c|}
			\hline
			Algorithm & Update time & Pattern matching time \\
			\hline \hline
			Idury \& Sch{\"a}ffer~\cite{idury1994dynamic} & $O(m(k d^{1/k} + \log{\sigma}))$ & $O(n(k+\log\sigma) + k \cdot \occ)$ \\
			\hline
			Amir \etal~\cite{amir1995improved} & $O(m(\frac{\log d}{\log\log d}+\log\sigma))$ & $O(n(\frac{\log d }{\log\log d} + \log\sigma) + \occ \frac{\log d}{\log \log d})$ \\
                        \hline
			Alstrup \etal~\cite{AlstrupHR98,AlstrupHR98tr} & $O(m\log\sigma + \log\log d)$ & $O(n(\frac{\log d }{\log\log d} + \log\sigma) + \occ)$ \\            
			\hline
			Chan \etal~\cite{ChanHLS07} & $O(m\log^2 d)$ & $O((n+ \occ)\log^2 d)$ \\
			\hline
			Hon \etal~\cite{hon2009succinct} & $O(m\log\sigma + \log d)$ & $O(n\log d + \occ)$ \\
			\hline
			Feigenblat \etal~\cite{FeigenblatPS17} & $O(\frac{1}{\epsilon}m\log d)$  & $O(n \log\log d \log\sigma + \occ)$ \\
			\hline
			Meyer$\dagger$~\cite{meyer1985incremental} & $O(l_{\max}\cdot d \cdot\sigma)$ & $O(n\log\sigma + \occ)$ \\
			\hline
			Tsuda \etal~\cite{tsuda1995incremental}  & $O(z\log\sigma)$ &  $O(n\log\sigma + \occ)$ \\
			\hline
			\hline
            DAWG based$\dagger$ & $O(m\log \sigma)$  & $O(n\log\sigma + \occ)$ \\
			\hline
			DAWG based & $O(\sigma m + \frac{\log d}{\log\log d})$  & $O(n (\frac{\log d}{\log\log d} + \log \sigma) + \occ\frac{\log d}{\log\log d})$ \\
			\hline
			AC-automaton based$\dagger$ & $O(m\log\sigma + \fnum + \onum)$ & $O(n\log\sigma + \occ)$ \\
			\hline
			AC-automaton based & $O(\sigma m + \fnum + \onum)$ & $O(n\log\sigma + \occ)$ \\
			\hline
		\end{tabular}	
	\end{center}
\end{table}

Remark that in all the above-mentioned approaches except Meyer's,
the pattern matching time to search a text for dictionary patterns is sacrificed to some extent.
Tsuda \etal~\cite{tsuda1995incremental} proposed
a dynamic dictionary matching algorithm, which follows and extends Meyer's method.
Whilst Tsuda \etal's method retains $O(n \log \sigma + \occ)$ pattern matching time, still it requires $O(z \log \sigma)$ time to update the AC-automaton upon each insertion/deletion, where $z$ is the size of the AC-automaton.
Note that in the worst case
this can be as bad as constructing the AC-automaton from scratch,
since $z$ can be as large as the dictionary size $d$.
Ishizaki and Toyama~\cite{ishizaki2012incremental} introduced a data structure called an \emph{expect tree} which efficiently updates the dictionary for insertion of patterns and showed some experimental results, but unfortunately no theoretical analysis were provided.
See Table~\ref{table:summary} for a summary of the update times and
pattern matching times for these algorithms.

Along this line, in this paper, we propose new efficient algorithms for the semi-dynamic and dynamic dictionary matching problems, where pattern matching can still be performed in $O(n \log{\sigma} + \occ)$ time. 

Firstly, we show a dynamic dictionary matching algorithm which is based on
Blumer \etal's \emph{directed acyclic word graphs} (\emph{DAWGs})~\cite{blumer1985smallest,Blumer:1987:CIF:28869.28873}.
The DAWG of a dictionary $D$ is a (partial) DFA of size $O(d)$
which recognizes the suffixes of the patterns in $D$.
We show how to perform dynamic dictionary matching with DAWGs,
by modifying Kucherov and Rusinowitch's algorithm
which originally uses DAWGs for pattern matching with variable length don't cares~\cite{Kucherov1997129}.
The key idea is to use efficient
\emph{nearest marked ancestor} (\emph{NMA}) data structures~\cite{westbrook92:_fast_increm_planar_testin,AlstrupHR98,AlstrupHR98tr}
on the tree induced from the suffix links of the DAWG.
For the semi-dynamic dictionary matching, 
our DAWG method achieves $O(m \log \sigma)$ update time,
$O(n \log \sigma + \occ)$ pattern matching time,
and uses linear $O(d)$ space,
where $m$ is the length of pattern $p$ to insert.
For the dynamic version of the problem,
our DAWG method uses $O(\sigma m + \log d / \log \log d)$
update time, $O(n (\log d / \log \log d + \log \sigma) + \occ \log d / \log \log d)$ time, and $O(d)$ space.
The term $\sigma m$ in the update time is indeed unavoidable
for maintaining the DAWG in the dynamic setting,
namely, we will also show that there is a sequence of
insertion and deletion operations for patterns of length $m$
such that each insertion/deletion operation takes $\Omega(\sigma m)$ time.

Secondly, we present another algorithm for the semi-dynamic and dynamic dictionary matching problem
which is based on the AC-automaton.
This is closely related to our first approach,
namely, the second algorithm updates the AC-automaton with the aid of the DAWG.
The algorithm uses $O(d)$ space,
finds pattern occurrences in the text in $O(n \log \sigma + \occ)$ time,
and 
updates the AC-automaton in $O(m\log\sigma + \fnum + \onum)$ time
with additional DAWG update time, $O(m\log\sigma)$ time for semi-dynamic and $O(\sigma m)$ time for dynamic settings,
where
$\fnum$ is the number of states whose failure link needs to be updated,
and $\onum$ is the number of states on which the value of the output function needs to be updated.
Therefore, when $\fnum$ and $\onum$ are sufficiently small and $\sigma$ is constant,
our update operation can be faster than other approaches.
Notice that 
$\onum$ is negligible unless the pattern $p$ to insert/delete
is a common prefix of many other patterns;
in particular $\onum = 0$ for any prefix codes.
Also, $\fnum$ is negligible unless
the prefixes of $p$ are common substrings of many other patterns.
In what follows, the update time of the AC-automaton refers to the time cost to update the AC-automaton after insertion/deletion of a pattern.
We emphasize that the update time of our algorithm is optimal
for constant-size alphabets,
since \emph{any} algorithm which explicitly maintains the AC-automaton
must use at least $\Omega(m + \fnum + \onum)$ time
to update the automaton.
Finally, we give tight upper and lower bounds on $\fnum$ and $\onum$
in the worst case.

A preliminary version of this work appeared in~\cite{DiptaramaYS16}.


\section{Preliminaries}

Let $\Sigma$ denote an \emph{alphabet} of size $\sigma$.
An element of $\Sigma^*$ is called a \emph{string}.
For a string $w$, the length of $w$ is denoted by $|w|$.
The \emph{empty string}, denoted by $\varepsilon$, is the string of length $0$.
For a string $w = x y z$, strings $x$, $y$, and $z$ are called \emph{prefix}, \emph{substring}, and \emph{suffix} of $w$, respectively.
For a string $w$, let $\Substr{w}$ denote the set of all substrings of $w$, and for a set $W = \{w_1, w_2, \ldots, w_r \}$ of strings,
let $\Substr{W} = \bigcup_{i=1}^{r} \Substr{w_i}$.
Similarly, let $\PrefixSet{W}$ be the set of all prefixes of strings in $W$.
For a string $w$, $w[i]$ denotes the $i$-th symbol of $w$ and $w[i:j]$ denotes the substring of $w$ that begins at position $i$ and ends at position $j$.

Let $D=\{p_1,p_2,\ldots,p_r\}$ be a set of \emph{patterns} over $\Sigma$,
called a \emph{dictionary}.
Let $d$ be the total length of the patterns in the dictionary $D$,
namely, $d = \sum_{i=1}^{r} |p_i|$.
The \emph{Aho-Corasick Automaton}~\cite{aho1975efficient} of $D$, denoted by $\ACautomaton{D}$, is a \emph{trie} of all patterns in $D$,
consisting of \emph{goto}, \emph{failure} and \emph{output} functions.
We often identify a state $s$ of $\ACautomaton{D}$
with the string obtained by concatenating all the labels found on the path from the root to the state $s$.
The state transition function \emph{goto} is defined so that 
for any two states $s, s' \in \PrefixSet{D}$ and any character $c \in \Sigma$, if $s' = s c$ then $s'=\goto{s}{c}$.
The failure function is defined by $\failure(s)=s'$ where $s'$ is the longest proper suffix of $s$ such that $s' \in \PrefixSet{D}$.
Finally, $\out(s)$ is the set of all patterns that are suffixes of $s$.
$\ACautomaton{D}$ is used to find occurrences of any pattern in $D$ on a text.
In this paper 
we omit the basic construction algorithm of $\ACautomaton{D}$ and
how it can be used to solve the dictionary pattern matching problem
(see~\cite{aho1975efficient,JewelsOfStringology2002} for details).

For any string $x$,
let
\[ \ePos{x}{D} = \{(i, j) \mid x = p_i[j-|x|+1:j], |x| \leq j \leq |p_i|, p_i \in D \},\]
namely, $\ePos{x}{D}$ represents the set of ending positions
of $x$ in patterns of $D$.
For any $x, y \in \Substr{D}$, we define the equivalence relation $\equiv_{D}$
such that $x \equiv_{D} y$ iff $\ePos{x}{D} = \ePos{y}{D}$.
We denote by $[x]_D$ the equivalence class of $x$ with respect to $\equiv_D$.
The \emph{directed acyclic word graph} (\emph{DAWG})~\cite{Blumer:1987:CIF:28869.28873,blumer1985smallest} of $D$, denoted by $\DAWG{D}$, is an edge-labeled directed acyclic graph $(V, E)$ such that
\begin{eqnarray*}
  V & = & \{ [x]_D \mid x \in \Substr{D} \}, \\
  E & = & \{([x]_D, c, [xc]_D) \mid x, xc \in \Substr{D}, c \in \Sigma, x \not \equiv_{D} xc \}.
\end{eqnarray*}
Namely, each node\footnote{To avoid confusion, we refer to a vertex in DAWGs as a \textit{node},
and a vertex in AC-automata as a \textit{state} in this paper.}
of $\DAWG{D}$ represents each equivalence class of substrings of $D$,
and henceforth we will identify a DAWG node with an equivalence class
of substrings.
The node $[\varepsilon]_D$ is called the \emph{source} of $\DAWG{D}$.
For each node $[x]_D$ except the source, the \emph{suffix link}
is defined by $\suf([x]_D) = [y]_D$, where $y$ is the longest suffix of $x$ satisfying $x \not \equiv_{D} y$.
For convenience, we define $\suf^{1}([x]_D) = \suf([x]_D)$
and $\suf^i([x]_D) = \suf(\suf^{i-1}([x]_D))$ for $i > 1$.
A node $v$ of $\DAWG{D}$ is called a \emph{trunk node}
if there is a path from the source to $v$
which spells out some prefix of a pattern in $\PrefixSet{D}$,
and it is called a \emph{non-trunk node} otherwise. 
An edge $e$ from $[x]_D$ to $[xc]_D$ labeled by $c$ is a \emph{primary edge} if both $x$ and $xc$ are the longest string in their equivalence classes,
otherwise it is a \emph{secondary edge}.
It is known (c.f.~\cite{Blumer:1987:CIF:28869.28873})
that the numbers of nodes, edges, and suffix links of $\DAWG{D}$
are all linear in $d$.
Fig.~\ref{fig:inverse} shows an example of a DAWG.

By the properties of AC-automata and DAWGs,
for each state $s$ in $\ACautomaton{D}$,  
there exists a unique node $v$ in $\DAWG{D}$ that corresponds to $s$, so that $\ACautomaton{D}$ can be consistently embedded into $\DAWG{D}$.
Because $s \in \PrefixSet{D}$, the corresponding node $v$ is a trunk node and each trunk node has its corresponding state.
Therefore, there exists a one-to-one mapping from the set of trunk nodes in $\DAWG{D}$ to the states of $\ACautomaton{D}$.
We denote this mapping by $s = \AC(v)$,
where $v$ is a DAWG trunk node and $s$ is the corresponding AC-automaton state.
We denote $v = \AC^{-1}(s)$ iff $s = \AC(v)$.
Fig.~\ref{fig:example} (a) and (b) show the AC-automaton and DAWG
of $D=\{{\tt abba},{\tt aca},{\tt cbb}\}$, respectively, 
where each number in nodes and states expresses the correspondence.

Our algorithm to follow will make a heavy use of the following lemma,
which characterizes the relationship between
the states of $\ACautomaton{D}$ and the trunk nodes of $\DAWG{D}$.
\begin{lemma}\label{lemma:slink and flink}
Let $s$ and $s'$ be any states in $\ACautomaton{D}$, and let
$v = \AC^{-1}(s)$ and $v' = \AC^{-1}(s')$ be corresponding trunk nodes in $\DAWG{D}$.
Then, $s' = \failure(s)$ iff there exists an integer $k \geq 1$ such that $v' = \suf^{k}(v)$, and when $k \geq 2$, $\suf^{i}(v)$ is a non-trunk node for all $1 \leq i < k$.
\end{lemma}

\begin{proof}
($\Longrightarrow$)
  Suppose $s' = \failure(s)$.
  Then, by definition, $s'$ is a proper suffix of $s$.
  Hence there exists an integer $k \geq 1$ such that
  $v' = \suf^{k}(v)$.
  Also, by definition, $k$ is the smallest such that $\slink^k(v)\in \PrefixSet{D}$.
  Hence, when $k \geq 2$, $\suf^{i}(v)$ is a non-trunk node for all $1 \leq i < k$.

($\Longleftarrow$)
  Suppose there exists an integer $k \geq 1$ such that $v' = \suf^{k}(v)$.
  When $k = 1$, clearly $s' = \failure(s)$.
  When $k \geq 2$ and $\suf^{i}(v)$ is a non-trunk node for all $1 \leq i < k$,
  then $k$ is the smallest integer such that $v' = \suf^{k}(v)$
  is a trunk node.
  Hence $s' = \failure(s)$.
\end{proof}


\section{Maintenance of Inverse Suffix Links of DAWG}\label{sec:invdawg}

\begin{figure}[t]
	\centering
	\includegraphics[scale=0.65]{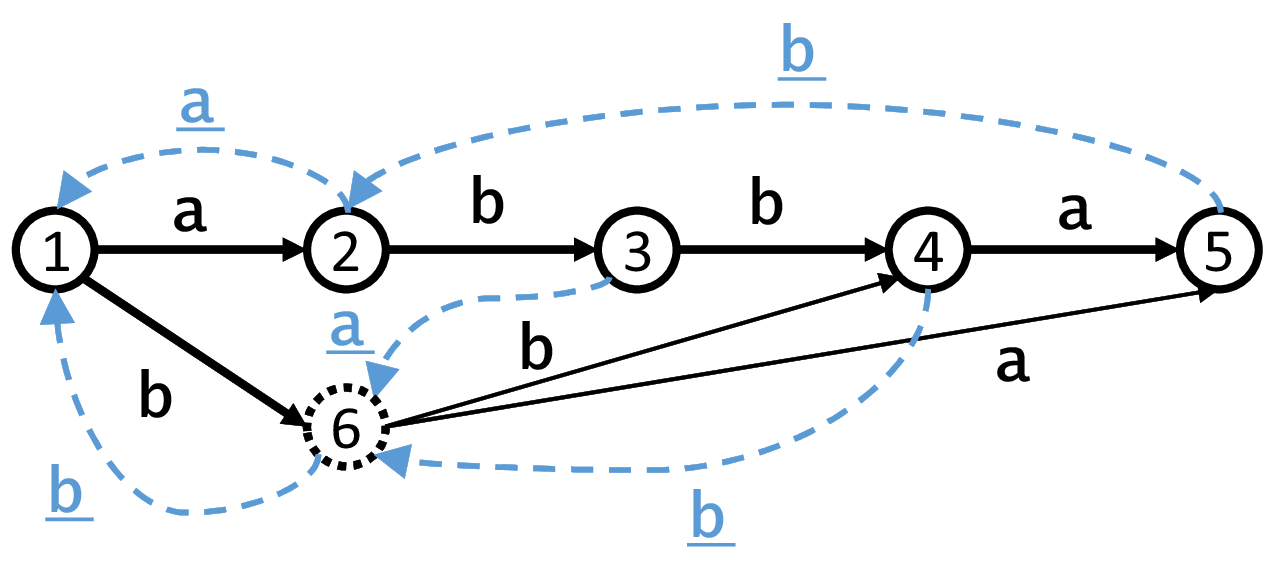}\\
	\caption{$\DAWG{\{\tt abba\}}$. Solid-line circles show trunk nodes and the dashed-line circle shows a non-trunk node.
		Thick solid lines, thin solid lines, and dashed lines show primary edges, secondary edges and suffix links, respectively.
	}
	\label{fig:inverse}
\end{figure}

\begin{figure}[t]
	\centering
	\begin{minipage}[t]{0.49\hsize}
		\centering
		\includegraphics[scale=0.45]{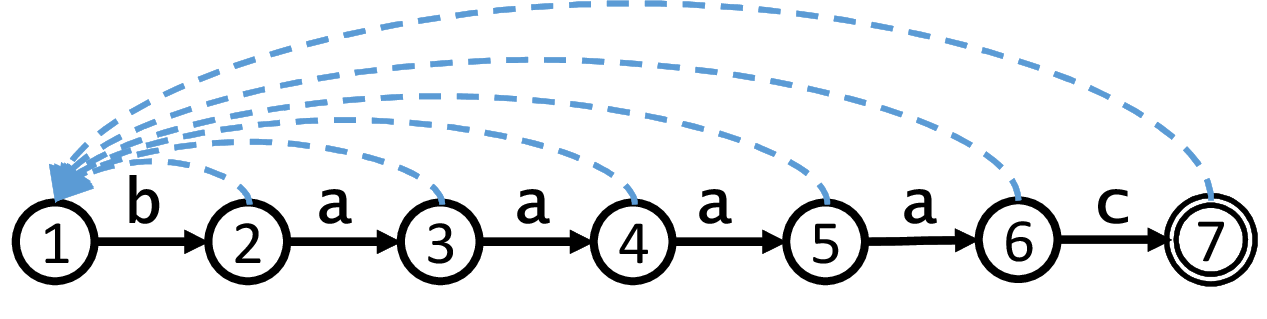}\\
		\ \ \ \scriptsize{(a)}
	\end{minipage}
	\begin{minipage}[t]{0.49\hsize}
		\centering
		\includegraphics[scale=0.45]{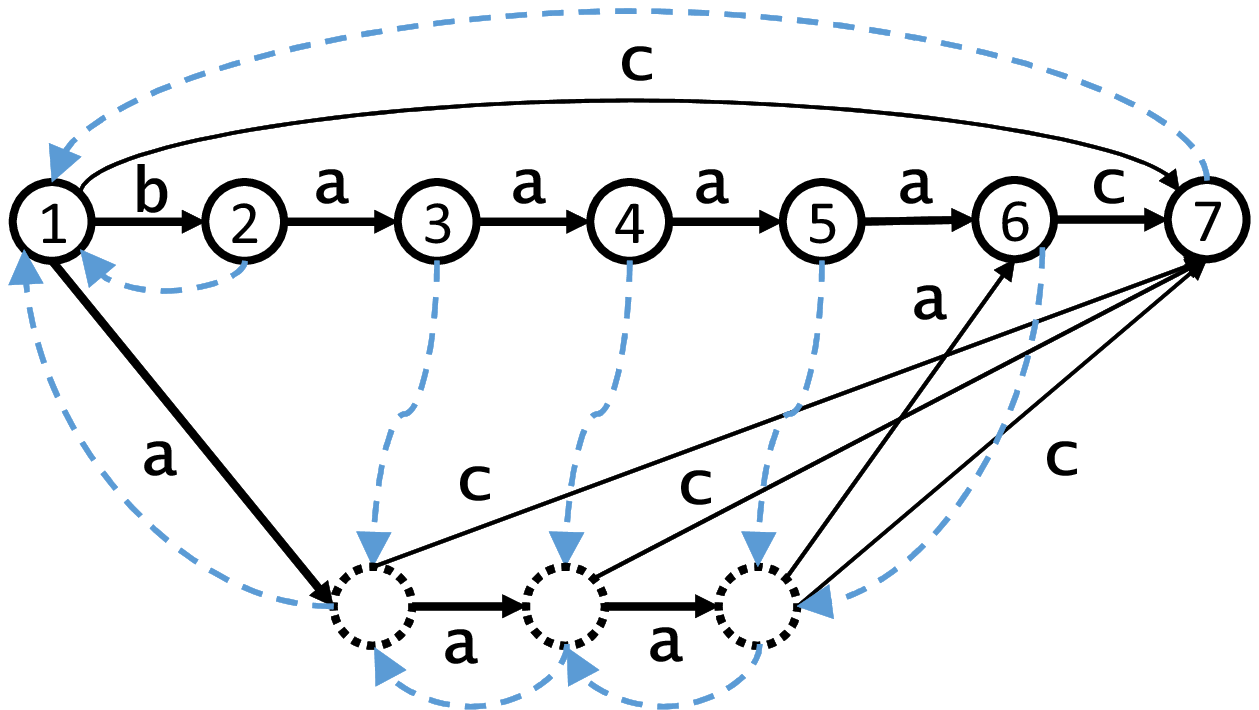}\\
		\ \ \  \scriptsize{(b)}
	\end{minipage}
	\caption{(a)$\ACautomaton{\{{\tt baaaac}\}}$ and (b) $\DAWG{\{{\tt baaaac}\}}$.}
	\label{fig:comparison}
\end{figure}

Meyer~\cite{meyer1985incremental} and Tsuda \etal~\cite{tsuda1995incremental} used the inverse of the failure function to update the AC-automaton.
Although the inverse failure function can be stored in a total of $O(d)$ space,
it is not trivial whether one can efficiently access and/or update
the inverse failure function,
because the number of inverse failure links of each state may change dynamically
and can be as large as the number of states in the AC-automaton.
For instance, let us consider $\ACautomaton{D}$ for $D = \{{\tt baaaac}\}$ over $\Sigma = \{ {\tt a}, {\tt b}, {\tt c}\}$ in Fig.~\ref{fig:comparison}~(a).
Its root is pointed by $6$ failure links.
When a new pattern ${\tt c}$ is inserted to $D$, then the above algorithms first create a new state $s$, a new transition from the root to $s$, and a failure link from $s$ to the root.
The real difficulty arises when they try to find which suffix links should be updated to point at $s$;
they must follow all the 6 inverse failure links from the root
and get $6$ states numbered $2$, $3$, \ldots, $7$, and check whether there is an edge labeled ${\tt c}$ from each of them, although only one state $7$ should be updated.
Ishizaki and Toyama~\cite{ishizaki2012incremental} introduced an auxiliary tree structure called an \emph{expect tree} to reduce the number of the candidates and showed some experimental results, but no theoretical analysis is provided.
Unfortunately, their algorithm behaves the same for the above example.
Therefore, maintaining the inverse failure links to update the AC-automaton might be inefficient.

In order to overcome this difficulty, we pay our attention to the suffix links of $\DAWG{D}$, instead of the failure links of $\ACautomaton{D}$.
It is known~(see, e.g.~\cite{blumer1985smallest, Blumer:1987:CIF:28869.28873, JewelsOfStringology2002, Crochemoretext}) that 
the inverse suffix links of all nodes in $\DAWG{D}$ form the suffix tree of the reversed patterns in $D$,
so that for any node $v$ in $\DAWG{D}$, 
each suffix link pointing at $v$ is labeled by a distinct symbol which is the first symbol of the edge label in the suffix tree.
Formally, the label of a suffix link is defined as follows.
Let $xy$ be the longest string in $[xy]_D$, $y$ be the longest string in $[y]_D$, and $\slink([xy]_D)=[y]_D$,
the label of this suffix link is $x[|x|]$.
In Fig.~\ref{fig:inverse}, the label of each suffix link is showed by an underlined symbol.
Therefore, the number of suffix links that point at $v$ is at most $\sigma$, and
the inverse suffix links can be accessed and updated in $O(\log \sigma)$ time using $O(d)$ total space.
This means that 
when a pattern of length $m$ is inserted to or deleted from the dictionary,
the inverse suffix links can be maintained in $O(m\log \sigma)$ time.

It is known that DAWGs
can be used for solving the pattern matching problem with a single pattern~\cite{Crochemore1988}.
However, it is not trivial to maintain the output function efficiently for dynamic and multiple patterns, as is pointed out by Kucherov and Rusinowitch~\cite{Kucherov1997129}.
In the next section,
we shall show our algorithm which efficiently maintains
the output function on the DAWG.


	\section{Dynamic Dictionary Matching by using DAWG}\label{sec:ddmdawg}

\begin{algorithm2e}[t]
\setlength{\baselineskip}{0mm}
\caption{Dynamic dictionary matching algorithm by using a DAWG}
\label{alg:ddmdawg}

\KwIn{A text string $T$.}
\KwOut{Occurence positions of all pattern in the dictionary.}

\SetKwProg{Fn}{Function}{}

$\activeNode = \rootNode$\;
\For{$1\leq i \leq n$}{
	\While{$\chk{\activeNode}{T[i]}$ $\mathbf{and}$ $\activeNode \ne \rootNode$}{
			$\activeNode = \suf(\activeNode)$\;
	}
	\If{$\chk{\activeNode}{T[i]}$}{
		$\activeNode = \trans(\activeNode,T[i])$\;
	}
	$\outNode = \activeNode$\;
	\If{$\outNode\mathrm{\ is\ marked}$}{
		$\Output(\outNode)$\;
	}
	\While{$\nma(\outNode) \ne \Null$}{
		$\outNode = \nma(\outNode)$\;
		$\Output(\outNode)$\;
	}
}

\Fn{$\chk{\node}{c}$}{
	\lIf{$node$ $\mathrm{is\ not\ a\ trunk\ node}$}{$\Return$ $\False$}
	\lIf{$\trans(\node,c) = \Null$}{$\Return$ $\mathbf{false}$}
	\lIf{$c\ \mathrm{is\ a\ secondary\ edge}$}{$\Return$ $\False$}
	\lIf{$\trans(\node,c)$ $\mathrm{is\ not\ a\ trunk\ node}$}{$\Return$ $\False$}
	$\Return\ \True$\;
}

\end{algorithm2e}

In this section, we will describe how to perform dynamic dictionary matching with the DAWG.
This algorithm is a simple modification of Kucherov and Rusinowitch's algorithm~\cite{Kucherov1997129} for matching multiple strings with variable length don't-care symbols.

First we will discuss the time complexity to update the DAWG in the semi-dynamic and dynamic settings.
As it was shown in~\cite{Blumer:1987:CIF:28869.28873} the DAWG can be updated in $O(m\log\sigma)$ \emph{amortized} time for an insertion of a pattern of length $m$ in the semi-dynamic setting.
For the dynamic setting, Kucherov and Rusinowitch~\cite{Kucherov1997129} gave an algorithm which deletes a pattern from the dictionary, and claimed that the update time for deletion and insertion is the same as in the semi-dynamic setting.
However, in what follows we show that this is not true when the alphabet size is super-constant. Namely, the number of edges to be constructed when we split a DAWG node can be amortized constant by the \emph{total length} of the input strings in the semi-dynamic setting, but this amortization argument does not hold in the dynamic setting.
That is, we obtain the following lower bound for updating the DAWG in the dynamic setting.
\begin{lemma}\label{lem:DAWGlow}	
	In the dynamic setting where both insertion and deletion of patterns are supported,
	there exists a family of patterns such that
	$\Omega(\sigma m)$ time is needed when updating the DAWG for insertion and deletion of each pattern.
\end{lemma}

\begin{proof}
	To show an $\Omega(\sigma m)$ lower bound,
	consider a pattern $p = (\mathtt{ba})^{\frac{m}{2}}$
	and an initial dictionary $D = \{(\mathtt{ab})^i\mathtt{a}^jc \mid 1 \leq i \leq \frac{m}{2},\ j\in\{0,1\},\ c \in \Sigma\setminus \{\mathtt{a},\mathtt{b}\}\}$
	of size $d = \Theta(\sigma m^2)$.
	We insert $p$ to the dictionary and update $\DAWG{D}$ into $\DAWG{D\cup \{p\}}$.
	In this case we need to split a node each time we read a symbol from $p$,
	and construct $\sigma - 2$ edges labeled by $c \in \Sigma \setminus \{\mathtt{a}, \mathtt{b}\}$ from the new node.
	Hence, we need to create $\Omega(\sigma m)$ edges when we update the DAWG.
	Moreover, the same computation time $\Omega(\sigma m)$ is required when we delete the same pattern $p$ from $D\cup\{p\}$ and update the DAWG.
	
	If we repeat this operation more than $m$ times by inserting and deleting $p$, 
	we cannot amortize the update cost by the size of the dictionary.
	Therefore, we need $\Omega(\sigma m)$ operations to update the DAWG when inserting or deleting a pattern.
\end{proof}

The above lower bound is tight, namely,
below we will show a matching upper bound for updating the DAWG in the dynamic setting.
\begin{lemma}\label{lem:DAWGup}		
	In the dynamic setting where both insertion and deletion of patterns are supported,
	the DAWG can be updated in $O(\sigma m)$ time for insertion and deletion of patterns.
\end{lemma}

\begin{proof}
To show the upper bound, we will evaluate the number of edges and suffix links that are traversed and/or created during the update.

Let us first consider the insertion operation.
Let $p$ be a pattern of length $m$ to be inserted to the dictionary.
Suppose that the prefix $p[1:i-1]$ of $p$ has already been inserted
to the DAWG for $1 \leq i \leq m$.
Let $v$ be the DAWG node that represents $p[1:i-1]$.
There are three cases for the next pattern character $p[i]$:
\begin{enumerate}
  \item[(1)] There is a primary out-going edge of $v$ labeled with $p[i]$.
    In this case no new edge or node is created, and it takes $O(\log \sigma)$ time to traverse this primary edge.
  \item[(2)] There is no out-going edge of $v$ labeled with $p[i]$.
    In this case, a new sink node and a new edge from $v$ to this new sink labeled with $p[i]$ are created. Then, the algorithm follows a chain of suffix links from $v$ and insert new edges leading to the new sink labeled with $p[i]$, until finding the first node which has an out-going edge labeled with $p[i]$.
  \item[(3)] There is a secondary out-going edge of $v$ labeled with $p[i]$.
    Let $u$ be the node that is reachable from $v$ via the edge labeled with $p[i]$. This node $u$ gets split into two nodes $u$ and $u'$, and at most $\sigma$ out-going edges of the original node $u$ are copied to $u'$.
\end{enumerate}
It is clear that Case (1) takes $O(\log \sigma)$ time per character.
At most $i$ new edges can be introduced in Case (2),
but it follows from~\cite{blumer1985smallest} that
the total number of suffix links that are traversed is $O(m)$
for all $m$ characters of $p$.
Hence, Case (2) takes $O(\log \sigma)$ amortized time per character.
It is clear that Case (3) takes $O(\sigma)$ time.
Overall, a pattern of length $m$ can be inserted to the DAWG
in $O(\sigma m)$ total time.

The deletion operation can also be performed in $O(\sigma m)$ time,
since the deletion can be done in the same complexity as insertion
by reversing the insertion procedure
(see also Kucherov and Rusinowitch's result~\cite{Kucherov1997129}). 
\end{proof}

Next, we will describe how to find the occurrences of the patterns in the text by using the DAWG.
We will use a nearest marked ancestor (NMA) data structure on
the inverse suffix link tree.
In the NMA problem on a rooted tree,
each node of the tree is either marked or unmarked.
The NMA query returns the nearest marked ancestor of a given query node $v$
in the tree, or returns $\Null$ if $v$ has no marked ancestor.
The semi-dynamic NMA problem allows for marking operation only,
while the dynamic NMA problem allows for both marking and unmarking operations.
New leaves can be added to the tree in both of the problems,
and existing leaves can be removed in the dynamic problem.
There is a semi-dynamic NMA data structure~\cite{westbrook92:_fast_increm_planar_testin} which allows for NMA queries,
marking unmarked nodes, and inserting new leaves in amortized $O(1)$ time each.
For the dynamic NMA problem,
there is a data structure which permits NMA queries and 
both marking and unmarking
operations in worst-case $O(\log t / \log \log t)$ time,
and inserting new leaves in amortized $O(1)$ time,
where $t$ is the size of the tree~\cite{AlstrupHR98,AlstrupHR98tr}.
Both of the data structures use $O(t)$ space
and $O(t)$ preprocessing time.

In our dictionary pattern matching algorithm using the DAWG,
we mark each node $v$ of the inverse suffix link tree 
iff $v$ is a DAWG node that represents a pattern in
the dictionary\footnote{Kucherov and Rusinowitch~\cite{Kucherov1997129} used
Sleator and Tarjan's link-cut tree data structure~\cite{SleatorT83}
to maintain a dynamic forest induced from the inverse suffix link tree.
Our important observation here is that essentially the same operations
and queries in this application can be more efficiently supported
with NMA data structures.}.
Now, for a given node $w$ in the DAWG,
we can find all patterns in the dictionary that are suffixes of $w$ by performing
NMA queries from $w$ on the inverse suffix link tree as follows.
If $w$ itself is marked, then we output it.
Then, we perform NMA queries in the inverse suffix link tree from $w$,
until we find a node that has no marked ancestor.
This allows us to skip all unmarked nodes in the path from $w$
to the node,
and we output all marked nodes found by NMA queries in this path.

Algorithm~\ref{alg:ddmdawg} shows a pseudo-code of our
algorithm for dynamic dictionary matching by using the DAWG.
The algorithm only uses the trunk nodes and primary edges to perform
pattern matching.
Therefore, when the algorithm reads a character $c$ from the text,
it checks whether or not there is a primary edge 
which is labeled with $c$ and leads to a trunk node by using a
function $\chk{\node}{c}$.
If there is no such node, then the algorithm
follows a chain of suffix links until it reaches a trunk node,
and then performs the same procedure as above.
Thus, the suffix links of the DAWG replaces
the failure links of the corresponding AC-automaton.
The correctness is immediately justified by Lemma~\ref{lemma:slink and flink}.
As soon as the algorithm finds a primary edge
which is labeled with $c$ and leads to a trunk node,
it checks whether there is an occurrence of any pattern in the dictionary
by using NMA queries from this destination trunk node, as described previously.
This procedure is used as a substitute for the output function of
the AC-automaton.

Consider inserting a new pattern $p$ to the dictionary.
If $v$ is the DAWG node which represents $p$,
then $v$ is newly marked in the inverse suffix link tree,
and $v$ is the only node that gets marked in this stage.
Hence, exactly one unmarked node gets marked per inserted pattern.
For the same reasoning, exactly one marked node gets unmarked
per deleted pattern.
To delete an existing pattern $p$ from the dictionary
and hence from the DAWG,
we can use Kucherov and Rusinowitch's algorithm~\cite{Kucherov1997129}
which takes $O(\sigma m)$ time due to Lemmas~\ref{lem:DAWGlow}
and~\ref{lem:DAWGup},
where $m$ is the length of $p$.

Overall, we obtain the following.
\begin{theorem}
In the semi-dynamic setting where only insertion of patterns is supported,
the DAWG-based algorithm supports insertion of patterns in 
$O(m \log \sigma)$ time and pattern matching in $O(n \log \sigma + \occ)$ time.

In the dynamic setting where both insertion and deletion of patterns are supported,
the DAWG-based algorithm supports insertion/deletion 
in $O(\sigma m + \log d / \log\log d)$ time
and pattern matching in $O(n (\log d / \log\log d + \log \sigma) + \occ \log d / \log\log d)$ time.
The size of both data structures is $O(d)$.
\end{theorem}

\begin{proof}
The update times and space requirements
of both of the semi-dynamic and dynamic versions
should be clear from Lemmas~\ref{lem:DAWGlow},~\ref{lem:DAWGup} and the above arguments.

For pattern matching, we need to perform at least one NMA query
each time a character from a text is scanned,
and need to perform an NMA query each time an occurrence of a pattern is found.
Hence, it takes $O(n \log \sigma + \occ)$ time for the semi-dynamic setting
and $O(n (\log d / \log\log d + \log \sigma) + \occ \log d / \log\log d)$ time for the dynamic setting.
\end{proof}

\section{AC-Automaton Update Algorithm}

\begin{figure}[t]
	\centering
	\begin{minipage}[t]{0.49\hsize}
		\centering
		\includegraphics[scale=0.43]{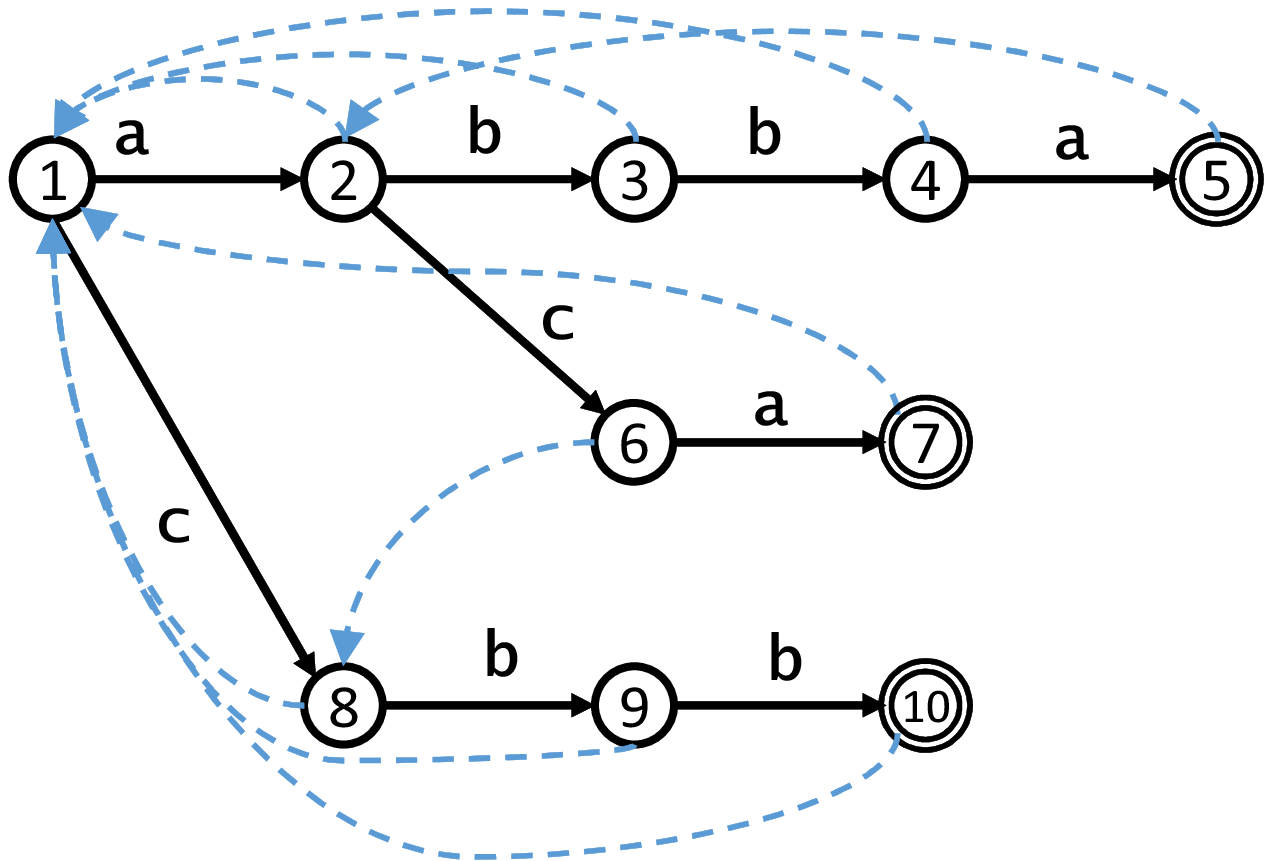}\\
		\ \ \ \scriptsize{(a)}
	\end{minipage}
	\begin{minipage}[t]{0.49\hsize}
		\centering
		\includegraphics[scale=0.43]{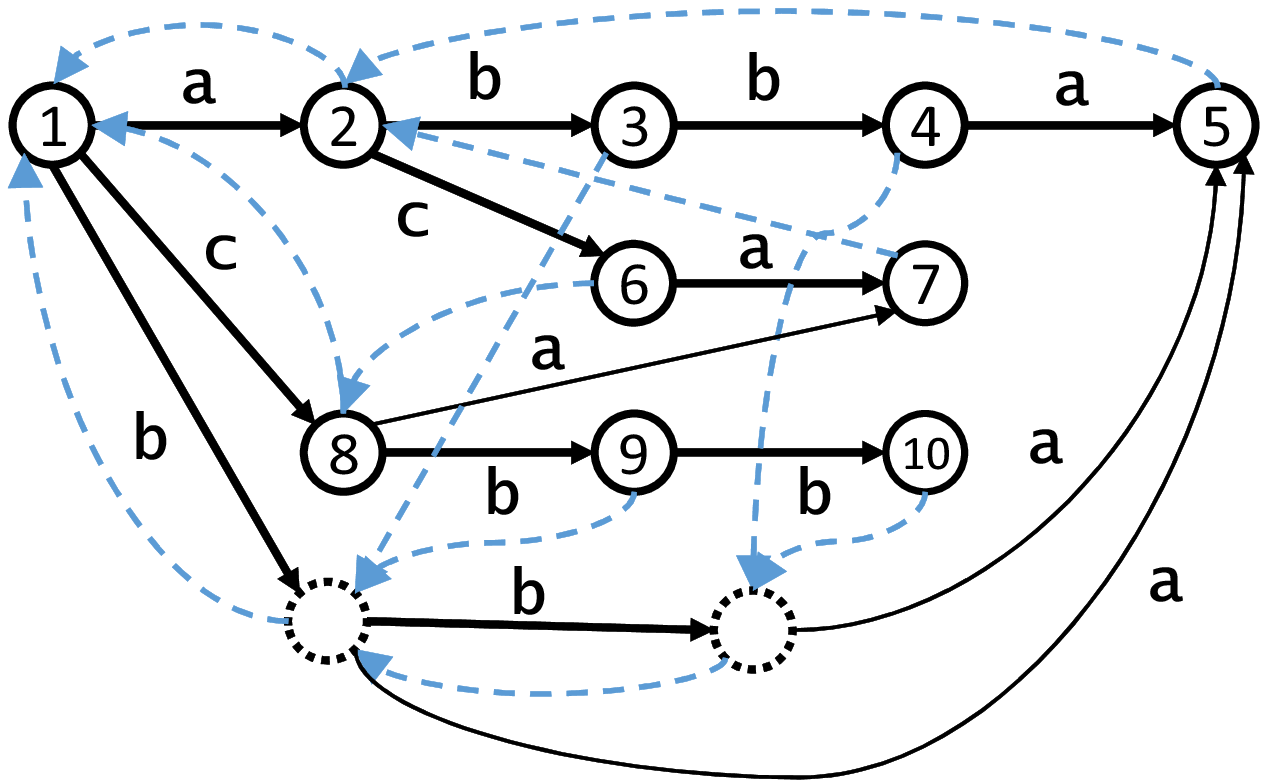}\\
		\ \ \  \scriptsize{(b)}
	\end{minipage}
	\caption{For a dictionary $D=\{{\tt abba},{\tt aca},{\tt cbb}\}$ (a) $\ACautomaton{D}$, and (b) $\DAWG{D}$.}
	\label{fig:example}
	\vspace{2mm}
	\centering
	\begin{minipage}[t]{0.49\hsize}
		\centering
		\includegraphics[scale=0.43]{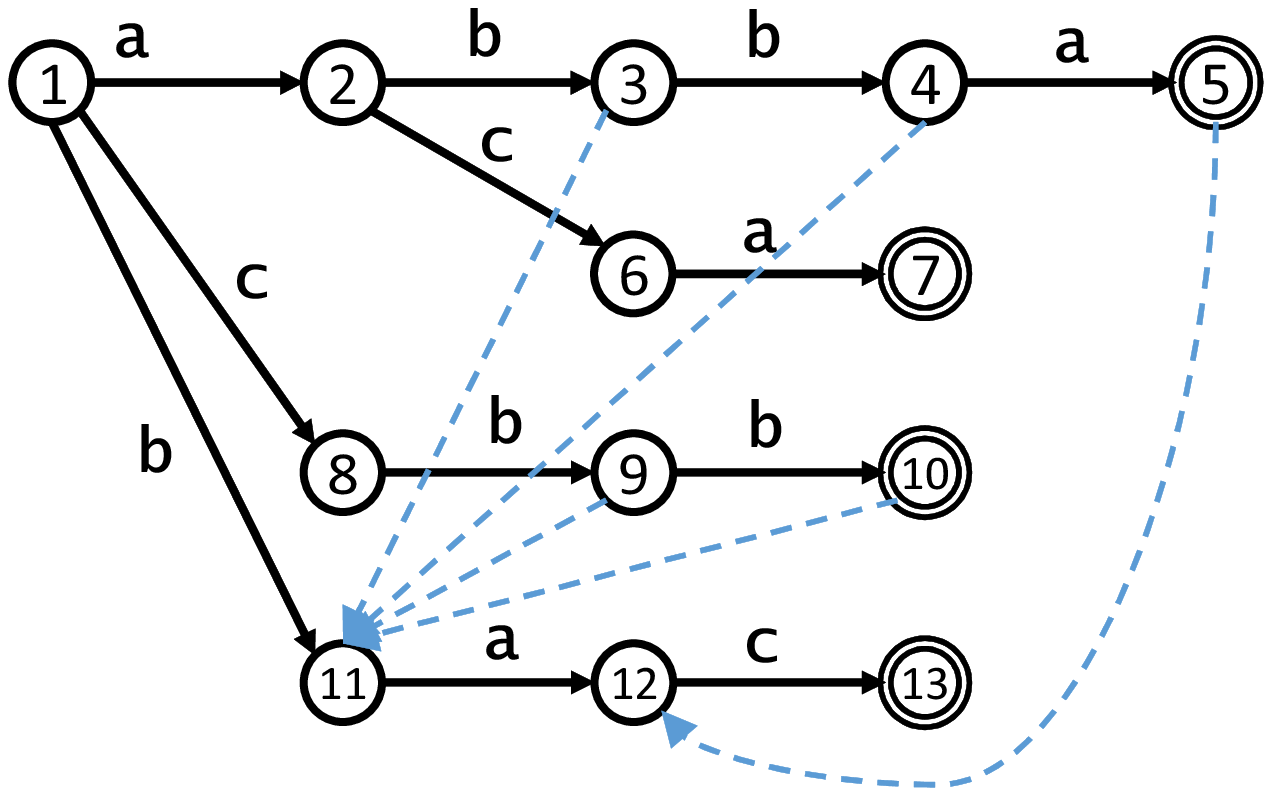}\\
		\ \ \ \scriptsize{(a)}
	\end{minipage}
	\begin{minipage}[t]{0.49\hsize}
		\centering
		\includegraphics[scale=0.43]{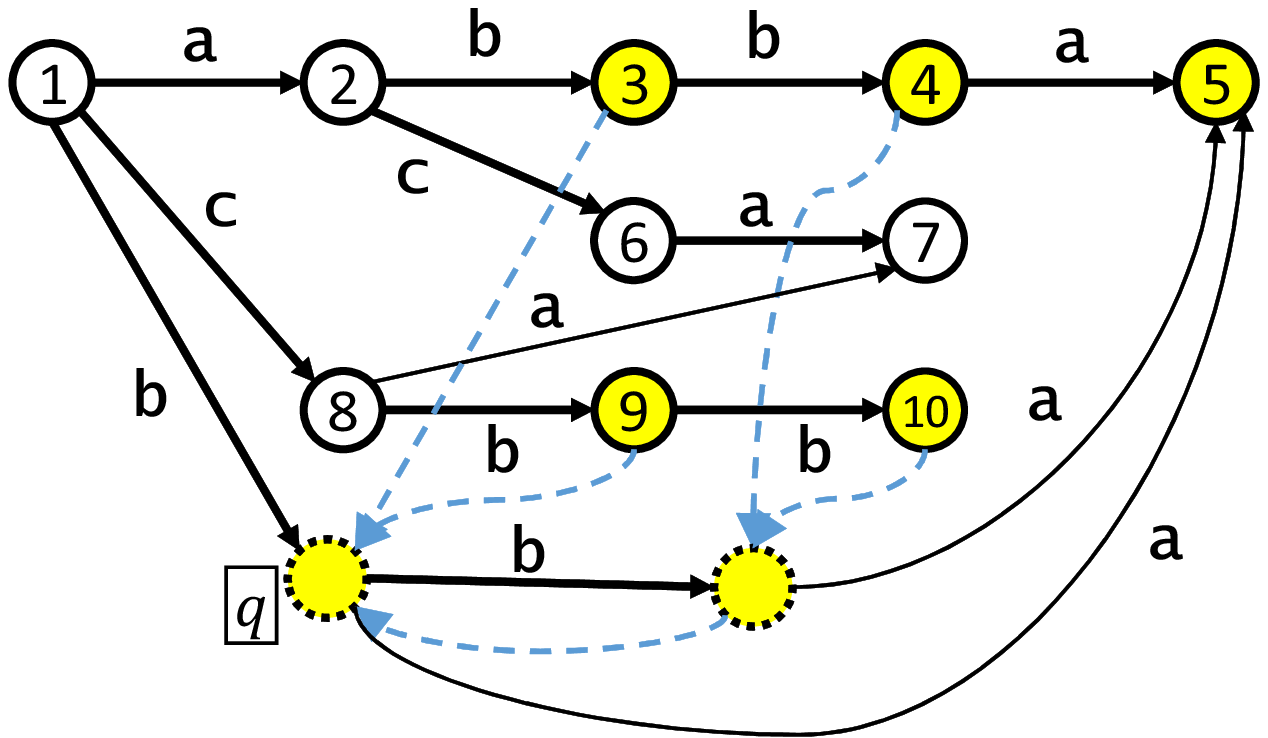}\\
		\ \ \  \scriptsize{(b)}
	\end{minipage}
	\caption{Illustration of updating process when inserting a pattern $p={\tt bac}$ into the dictionary $D=\{{\tt abba},{\tt aca},{\tt cbb}\}$.
		Compare them with Fig.~\ref{fig:example}.
		(a) The updated automaton $\ACautomaton{D'}$, where the only updated failure links are shown.
		(b) In $\DAWG{D}$, the only suffix links that are used for the update are shown,
		and the visited nodes are colored.
		}
	\label{fig:insertion_example}
\end{figure}

In this section we will describe how to perform dynamic dictionary matching
by using the AC-automaton and the DAWG for the dictionary.
Our algorithm performs pattern matching
in exactly the same manner as the original AC-algorithm,
while updating the AC-automaton
dynamically with the aid of the DAWG
upon insertion/deletion of patterns.
We will describe how to modify the AC-automaton by using the DAWG.
Note that we can simulate the AC-automaton with the DAWG augmented with the output function.
However, we will explicitly use the AC-automaton since
it makes the pattern matching algorithm simpler.

\subsection{Pattern insertion algorithm}

\begin{algorithm2e}[t]
\setlength{\baselineskip}{0mm}
\caption{$\getOutStates(p)$}
\label{alg:outstateins}
\KwOut{The states on which the output function should be updated}

$\outStates = \EmptySet$; $\activeNode = \rootNode$\;
\For{$1\leq i \leq m$ $\mathbf{and}$ $\activeNode \neq \Null$ \label{line:begin traverseP}}{
	$\activeNode = \trans(\activeNode,p[i])$\;\label{line:end traverseP}
} 

\If{$\activeNode \neq \Null$}{ \label{line:begin findOut}
	$\queue = \EmptySet$\;
	push $\activeNode$ to $\queue$\;
	\While{$\queue \neq \EmptySet$}{
		pop $\node$ from $\queue$\;
		\If{$\node$ $\mathrm{is\ a\ trunk\ node}$}{
			$\outStates = \outStates \cup \{\AC(\node)\}$
		}
		\For{$\lnode \in \isuf(\node)$}{
			push $\lnode$ to $\queue$\;  \label{line:end findOut}
		}
	}
}
$\mathbf{return}$ $\outStates$\; 
\end{algorithm2e}

\begin{algorithm2e}[t]
\setlength{\baselineskip}{0mm}
\caption{$\getFailStates(p,start)$}
\label{alg:failstateins}
\KwOut{A stack contains states whose failure link should be updated.}

$\stack = \EmptySet$; $\activeNode = \rootNode$\;
\For{$1\leq i \leq m$ $\mathbf{and}$ $\activeNode \neq \Null$ \label{line:fail1}}{
	$\activeNode = \trans(\activeNode,p[i])$\;
	\If{$i \geq \start$ $\mathbf{and}$ $\activeNode \neq \Null$}{
		push $(\activeNode,i)$ to $\stack$\;\label{line:fail2}
	}
}

\While{$\stack \neq \EmptySet$ \label{line:fail3}}{
	pop $(\activeNode,i)$ from $\stack$\;
	$\queue = \EmptySet$\;
	push $\activeNode$ to $\queue$\;
	\While{$\queue \neq \EmptySet$}{
		pop $\node$ from $\queue$\;
		\If{$\node$ $\mathrm{is\ a\ trunk\ node}$}{
			push $(\AC(\node),i)$ to $\failStates$\;
		}
		\If{$\node$ $\mathrm{is\ not\ marked}$}{
			mark $node$\;
			\If{$\node$ $\mathrm{is\ a\ branch\ node}$}{
				\For{$\lnode \in \isuf(\node)$}{
					push $\lnode$ to $\queue$\;\label{line:fail4}
				}
			}
		}
	}
}

$\mathbf{return}$ $\failStates$\;
\end{algorithm2e}

\begin{algorithm2e}
	\setlength{\baselineskip}{0mm}
	\caption{Pattern insertion algorithm of AC-automaton}
	\label{alg:insertion}
	\KwIn{new pattern $p$}
	\SetKwFor{Fn}{Function}{}
	
	$\activeState = \rootState$\;
	$\newStatesSet = \EmptySet$\;
	$\weight{\activeState} = \weight{\activeState} + 1$\;
	\For{$1\leq i \leq m$}{
		\lIf{$\goto{\activeState}{p[i]} \neq \fail$}{
			$\activeState = \goto{\activeState}{p[i]}$}
		\Else{
			create $\newState$\;
			$\goto{\activeState}{p[i]} = \newState$\;
			$\activeState = \newState$\;
			$\newStatesSet = \newStatesSet \cup \{\newState\}$\;
		}
		$\weight{\activeState} = \weight{\activeState} + 1$\;
		\If{i = m}{
			$\out(newState) = \out(\newState) \cup \{p\}$\;
		}
	} 
	
	$\failStates = \underline{\getFailStates}(p,m - |\newStatesSet| + 1)$\;
	\While{$\failStates \neq \EmptySet$}{
		pop $(\state,i)$ from $\failStates$\;
		$\failure(\state) = \newStatesSet[i-|\newStatesSet| + 1]$\;
	}

	$\activeState = \rootState$\;
	\For{$1\leq i \leq m$}{
		\If{$\goto{\activeState}{p[i]} \in \newStatesSet$}{
			$\failureState = \failure(\activeState)$\;
			\While{$goto(\failureState,p[i]) = \fail$}{
				$\failureState = \failure(\failureState)$\;
			}
			$\activeState = \goto{\activeState}{p[i]}$\;
			$\failure(\activeState) = \failureState$\;
			$\out(\activeState) = \out(\activeState) \cup \out(\failureState)$\;
		}
		\Else{
			$\activeState = \goto{\activeState}{p[i]}$\;
		}
	}
	$\outStates = \underline{\getOutStates}(p)$\;
	\lFor{$\state \in \outStates$}{
		$\out(\state) = \out(\state) \cup \{p\}$
	}
	
\end{algorithm2e}

We consider inserting a new pattern $p$ of length $m$ into the dictionary $D$,
and we denote the new dictionary by $D'= D \cup \{p\} = \{p_1,p_2,\ldots,p_r,p\}$.
It is known that $\DAWG{D}$ can be constructed in $O(d\log\sigma)$ time,
and can be updated to $\DAWG{D'}$ online in $O(m\log\sigma)$
amortized time~\cite{Blumer:1987:CIF:28869.28873}.
We update $\ACautomaton{D}$ to $\ACautomaton{D'}$ by using $\DAWG{D}$,
and then update $\DAWG{D}$ to $\DAWG{D'}$.
We also add $\weight{v}$ to each state $v$ of $\ACautomaton{D}$
that is the number of occurrences $v$ as prefix in $D$.
We will use $\weight{v}$ as a reference counter to determine whether $v$ should be deleted or not in the deletion algorithm later.

The key point of our algorithm is to update the output and failure functions of $\ACautomaton{D}$
in linear time with respect to the number of states that should be modified.
The $\mathit{goto}$ function can be updated easily by adding a new transition for a new state in the same way as in the AC-automaton construction algorithm.
We then update the output and failure functions efficiently by using inverse suffix links of $\DAWG{D}$.
Algorithm~\ref{alg:insertion} updates $\ACautomaton{D}$ when a new pattern is inserted to $D$, and
Algorithms~\ref{alg:outstateins} and~\ref{alg:failstateins} find the states on which the output and failure functions should be updated, respectively.

For any node $v$ in $\DAWG{D}$,
let $\isuf(v) = \{ x \mid \suf(x)=v \}$ be the set of its inverse suffix links.
The set $\isuf(v)$ for each $v$ is stored in an ordered array $v_a$ as described in Section~\ref{sec:invdawg}.
For the new pattern $p$,
we can divide $p$ to $p = xyz$ and categorize the prefixes of $p$ into three categories, so that for any $i,j,k$ with $1 \leq i \leq |x| < j \leq |x| + |y| < k \leq m$;
\begin{enumerate}
	\item $p[1:i]$ exists both in $\ACautomaton{D}$ and $\DAWG{D}$,
	\item $p[1:j]$ does not exist in $\ACautomaton{D}$ but exists in $\DAWG{D}$, and
	\item $p[1:k]$ exists in neither $\ACautomaton{D}$ nor $\DAWG{D}$.
\end{enumerate}
To update both output and failure functions of $\ACautomaton{D}$ to $\ACautomaton{D'}$ we only use nodes in $\DAWG{D}$ that represent prefixes in the second category.
Algorithm~\ref{alg:outstateins} follows inverse suffix links of a node representing $p$ recursively in $\DAWG{D}$,
in order to find all the states in $\ACautomaton{D}$ on which the output function needs to be updated.
On the other hand, Algorithm~\ref{alg:failstateins} follows inverse suffix links of nodes that represent $p[i:j]$ for $|x| < j \leq |x| + |y|$ (category 2) recursively,
until it reaches a trunk node $u$, and then saves the state $s = \pi(u)$ that corresponds to the trunk node to update its failure link later.

Fig.~\ref{fig:insertion_example} illustrates an example, 
where we insert a pattern $p={\tt bac}$ into the dictionary $D=\{{\tt abba},{\tt aca},{\tt cbb}\}$.
First, we create new states $11$, $12$, and $13$.
The string ${\tt b}$ is represented by node $q$ in $\DAWG{D}$, and by the new state $11$ in $\ACautomaton{D'}$,
thus there is at least one state whose failure link should be updated to point at the state $11$.
We will explain how to find these states below.
Similarly, we know that at least one failure link should be updated to point at the state $12$,
because the string ${\tt ba}$ represented by the state $12$ in $\ACautomaton{D'}$ is also represented by node $5$ in $\DAWG{D}$.
However, the string ${\tt bac}$, which is represented by the new state $13$, is not represented in $\DAWG{D}$,
thus we know that there is no state whose failure link should be updated to state $13$.
As a result, we have the set $\{11, 12\}$ of states.
(Lines~\ref{line:fail1}--\ref{line:fail2} in Algorithm~\ref{alg:failstateins})

We now explain how to find states whose failure links should be updated.
We begin by the deepest state in $\{11,12\}$, that is, state $12$.
We search the states from node $5$ in $\DAWG{D}$, which represents the same string ${\tt ba}$ as state $12$ in $\ACautomaton{D'}$.
When searching from node $5$, we do not search further because node $5$ is a trunk node.
Therefore, we update the failure link of state $5$ to state $12$.
Next, to find states whose failure links should be updated to state $11$,
we search the states from node $q$ in $\DAWG{D}$, which represents the same string ${\tt b}$ as state $11$ in $\ACautomaton{D'}$.
By following the inverse suffix links recursively from node $q$ until reaching a trunk node,
we get the set $\{3,4,9,10\}$ of trunk nodes (see Fig.~\ref{fig:insertion_example} (b)).
Therefore, we update the failure links of states $3$, $4$, $9$, and $10$ to state $11$.
(Lines~\ref{line:fail3}--\ref{line:fail4})

\subsection{Pattern deletion algorithm}

\begin{algorithm2e}[t]
\setlength{\baselineskip}{0mm}
\caption{Pattern deletion algorithm of AC-automaton}
\label{alg:deletion}
\KwIn{A pattern $p$}
\SetKwFor{Fn}{Function}{}

$\activeState = \rootState$\;
$\deleteStatesSet = \EmptySet$\;
\For{$1\leq i \leq m$}{
	$\activeState = \goto{\activeState}{p[i]}$\;
	\If{$\weight{\activeState} = 0$}{
		$\deleteStatesSet = \deleteStatesSet \cup \activeState$\;
	}
	\lElse{
		$\weight{\activeState} = \weight{\activeState} - 1$}
}

$\failStates = \underline{\getFailStates}(p,m - |\deleteStatesSet| + 1)$\;
\While{$\failStates \neq \EmptySet$}{
	pop $(\state,i)$ from $\failStates$\;
	$\failure(\state) = \failure(\newStatesSet[i-(m-|\newStatesSet|) + 1])$\;
}

$\activeState = \rootState$\;
$\outStates = \underline{\getOutStates}(p)$\;
\lFor{$\state \in \outStates$}{
	$\out(\state) = \out(\state) \setminus \{p\}$}

\For{$\state \in \deleteStatesSet$}{
		delete $\state$\;
}

\end{algorithm2e}

We consider deleting a pattern $p_i$ of length $m$ from the dictionary $D$,
and we denote the new dictionary by $D'= D \setminus \{p_i\} = \{p_1,\ldots,p_{i-1},p_{i+1},\ldots,p_r\}$.
Similarly to insertion, we can delete a pattern from $\DAWG{D}$ in $O(m\log\sigma)$ time~\cite{Kucherov1997129}.
We update $\ACautomaton{D}$ to $\ACautomaton{D'}$ by using $\DAWG{D}$,
and then update $\DAWG{D}$ to $\DAWG{D'}$.
The proposed deletion algorithm also updates the output and failure functions of $\ACautomaton{D}$
in linear time with respect to the number of states that should be modified.

Algorithm~\ref{alg:deletion} shows the proposed deletion algorithm.
First, the algorithm finds which states should be deleted.
The algorithm finds the states by decreasing the weight of states which represent prefixes of $p$.
The algorithm will delete the states whose weight becomes zero,
which means those states do not represent any prefix of any pattern in $D'$.

After the algorithm has found the states which should be deleted,
it will update the states whose failure links should be updated.
A state should be updated if the failure link of the state is pointing at one of the nodes that will be deleted.
Such states can be found by traversing reverse failure links of the states.
From Lemma~\ref{lemma:slink and flink} we can use inverse suffix links of the DAWG instead of inverse failure links of the AC-automaton to find the states.
The algorithm uses $\getFailStates(p,start)$ in Algorithm~\ref{alg:failstateins} to find the states and
update them from the states of which the suffix links point to shallower states.

Next, the algorithm will update the output function of the AC-automaton.
The output function of a state should be updated if and only if $p$ is a suffix of the string that is represented by the state.
The algorithm uses $\getOutStates(p)$ in Algorithm~\ref{alg:outstateins} to find the states whose output function should be updated.
Last, the algorithm will delete the respective states.

\subsection{Correctness of the algorithms}

We now show the correctness of Algorithms~\ref{alg:outstateins} and \ref{alg:failstateins}.

\begin{lemma}
	\label{lem:outstate}
	Algorithm~\ref{alg:outstateins} \emph{correctly} returns the set of states on which output functions should be updated.
\end{lemma}

\begin{proof}
	When a new pattern $p$ is inserted to a dictionary $D$, we have to update the output function of every state $s$ in $\ACautomaton{D}$ 
	such that $p$ is a suffix of the string $s$.
	If there is no node in $\DAWG{D}$ representing $p$, we know that no such a string $s$ exists in $D$.
	Otherwise, let $s_p$ be a new state created in $\ACautomaton{D'}$ to represent the pattern $p$.
	The output function of some state $s$ should be updated if and only if $s_p$ is reachable from $s$ via a chain of failure links.
	From Lemma~\ref{lemma:slink and flink}, for nodes $u=\AC^{-1}(s)$ and $v_p=[p]_D$, we have $v_p = \suf^i(u)$ for some $i$.
	Therefore, $s=\AC(u)$ can be found by following inverse suffix links from $v_p$ recursively.
\end{proof}

\begin{lemma}
	\label{lem:failstate}
	Algorithm~\ref{alg:failstateins} correctly returns the set of states whose failure links should be updated.
\end{lemma}

\begin{proof}
	By arguments similar to the proof of Lemma~\ref{lem:outstate},
	all the states that should be updated are reachable via chains of inverse suffix links from the nodes in $\DAWG{D}$ that correspond to the new states in $\ACautomaton{D'}$.
	Next, we will show that Algorithm~\ref{alg:failstateins} only returns the set $S$ of the states that should be updated.
	Let $x$ be a new state and 
	$t=[x]_D$ be a node that represents the string $x$.
	Assume that $S$ contains a state $s$ that can be reached by following inverse failure links from $x$ recursively
	but should not be updated.
	Let $u = \AC^{-1}(s)$ and $v = \AC^{-1}(\failure(s))$ be trunk nodes in $\DAWG{D}$ corresponding to $s$ and $\failure(s)$, respectively.
	From Lemma~\ref{lemma:slink and flink},
	$v = \suf^i(u)$ and $t = \suf^j(v)$ for some $i$ and $j$. 
	Since Algorithm~\ref{alg:failstateins}, started from $t$, stops a recursive search after reaching a trunk node ($v$ in this case),
	it would not find $u$.
	Therefore, $s=\AC(u) \not\in S$.
\end{proof}


\section{Algorithm Complexity Analysis}

We now show the time complexity of Algorithms~\ref{alg:outstateins} and \ref{alg:failstateins}.

\begin{lemma}[\cite{blumer1985smallest}] \label{lem:lognest member of [x]_D}
	A string $x \in \Substr{D}$ is the longest member of $[x]_D$ if and only if either $x \in \PrefixSet{D}$ or $ax, bx \in \Substr{D}$ for some distinct $a, b \in \Sigma$. 
\end{lemma}

\begin{lemma}
\label{lem:suffixlink}
	For any non-trunk node in DAWG, there exist at least two suffix links that point at it.
\end{lemma}

\begin{proof}
	Let $[x]_D$ be any non-trunk node in $\DAWG{D}$ and $x \in \Substr{D}$ be the longest member of $[x]_D$.
	Then $x \not\in \PrefixSet{D}$ because $[x]_D$ is a non-trunk node.
	By Lemma~\ref{lem:lognest member of [x]_D}, there exist two distinct $a, b \in \Sigma$ such that
	$ax, bx \in \Substr{D}$. Because $x$ is the longest member of $[x]_D$, we have $[ax]_D \neq [x]_D$.
	Thus, $\suf{([ax]_D)} = [x]_D$ because $x$ is a suffix of $ax$.
	Similarly, $\suf{([bx]_D)} = [x]_D$. Because $[ax]_D \neq [bx]_D$, the non-trunk node $[x]_D$ is pointed by at least two suffix links.
\end{proof}

\begin{lemma}
\label{lem:outstatetime}
	Algorithm~\ref{alg:outstateins} runs in $O(m\log\sigma+\onum)$ time,
	where $\onum$ is the number of states on which output function should be updated.
\end{lemma}

\begin{proof}
	At first, Algorithm~\ref{alg:outstateins} finds the node $v$ representing the pattern $p$, by traversing the nodes from the root, in Lines~\ref{line:begin traverseP}--\ref{line:end traverseP}.
	It takes $O(m\log\sigma)$ time.
	If it failed, done.
	Then we analyze the running time consumed in Lines~\ref{line:begin findOut}--\ref{line:end findOut} by counting the number $\ell$ of visited nodes in $\DAWG{D}$.
	These nodes form a tree, rooted at $v$ and connected by inverse suffix links chains.
	Let $\bnum$ (resp. $\tnum$) be the number of non-trunk (resp. trunk) nodes in this tree,
	and let $q$ be the number of nodes (either non-trunk or trunk) that are child nodes of some non-trunk node.
	Because every non-trunk node has at least two child nodes by Lemma~\ref{lem:suffixlink},
	we have $2\bnum \leq q$, and obviously $q \leq \bnum + \tnum$. Thus, $\bnum \leq \tnum$, which yields that
	$\ell = \bnum + \tnum \leq 2\tnum = 2|\outStates| = 2 \onum$.
	Therefore, Algorithm~\ref{alg:outstateins} runs in $O(m\log\sigma+\onum)$ time.
\end{proof}

\begin{lemma}
\label{lem:failstatetime}
	Algorithm~\ref{alg:failstateins} runs in $O(m\log\sigma+\fnum)$ time,
	where $\fnum$ is the number of states whose failure links should be updated.
\end{lemma}

\begin{proof}
	At first, Algorithm~\ref{alg:failstateins} finds the set $V$ of nodes representing the pattern $p[1:j]$ for $1 \leq j \leq m$
	such that $p[1:j]$ does not exist in $\ACautomaton{D}$ but does exist in $\DAWG{D}$,
	by traversing the nodes from the root, in Lines~\ref{line:fail1}--\ref{line:fail2}.
	The algorithm saves the nodes in a stack, because the algorithm will search from the deepest node.
	This takes $O(m\log\sigma)$ time.
	Then we analyze the running time consumed in Lines~\ref{line:fail3}--\ref{line:fail4} by counting the number $\ell$ of visited nodes in $\DAWG{D}$.
	These nodes form a forest, where each tree is rooted by some node in $V$ and connected by inverse suffix link chains,
	where some node in $V$ can be an inner node of a tree rooted by another in $V$.
	In this case, we mark the nodes that have been visited, so each node is visited at most twice.
	Let $\bnum$ (resp.\ $\tnum$) be the number of non-trunk (resp.\ trunk) nodes in this forest,
	and let $q$ be the number of nodes (either non-trunk or trunk) that are child nodes of some non-trunk node.
	Because every non-trunk node has at least two child nodes by Lemma~\ref{lem:suffixlink},
	we have $2\bnum \leq q$, and obviously $q \leq \bnum + \tnum$. Thus, $\bnum \leq \tnum$, which yields that
	$\ell = \bnum + \tnum \leq 2\tnum = 2|\failStates| = 2 \fnum$.
\end{proof}

\begin{theorem}
	AC-automaton can be updated for each pattern in $O(m\log\sigma + \fnum + \onum)$ time.
\end{theorem}

\begin{proof}
	The goto, failure and output functions of newly created states can be calculated in $O(m\log\sigma)$,
	similarly to the original AC-automaton construction algorithm.	
	From Lemmas~\ref{lem:outstatetime} and~\ref{lem:failstatetime},
	the output and failure functions on existing states can be updated in $O(m\log\sigma+\onum)$ and $O(m\log\sigma+\fnum)$, respectively.		
	Therefore, AC-automaton can be updated in $O(m\log\sigma + \fnum + \onum)$ time in total.
\end{proof}

Note that \emph{any} algorithm
which explicitly updates the AC-automaton 
requires at least $\Omega(m + \fnum + \onum)$ time.
Hence, the bound in the above theorem is optimal
except for the term $\log \sigma$ which can be ignored for
constant alphabets.
As was stated in Section~\ref{sec:intro},
$\fnum$ and $\onum$ can be considerably small in several cases.

The remaining question is
how large $\fnum$ and $\onum$ can be in the worst case.
The next theorem shows matching upper and lower bounds
on $\fnum$ and $\onum$.

\begin{theorem}
  For any pattern of length $m$,
  $\fnum = O(km)$ and $\onum = O(km)$, where $k$ is the number of patterns to insert to the current dictionary.
  Also, there exists a family of patterns for which $\fnum = \Omega(km)$ and $\onum = \Omega(km)$.
\end{theorem}

\begin{proof}
  In this proof, we only show bounds for $\fnum$; however,
  the same bounds for $\onum$ can be obtained similarly.
  
  First, we show an upper bound $\fnum = O(km)$.
  We begin with an empty dictionary 
  and insert patterns to the dictionary.
  Let $d$ be the total length of the patterns in the dictionary
  after adding all patterns,
  and let $\utotal$ be the total number of AC-automaton states whose failure links need to be updated during the insertion of all patterns.
  If $k$ is the number of patterns to insert,
  then clearly $\utotal \leq kd$ holds.
  Hence, the number of failure links to update
  per character is $\utotal /d \leq k$.
  This implies that for any pattern of length $m$,
  the number $\fnum$ of failure links to update is $O(km)$.

  To show a lower bound $\fnum = \Omega(km)$,
  consider an initial dictionary $D = \{c_i a^k \mid 1 \leq i \leq x\}$
  of $x \geq 1$ patterns, where $k \geq 1$ and
  $c_i \neq c_j$ for any $1 \leq i \neq j \leq x$.
  For each $j = 1, 2, \ldots, k$ in increasing order,
  we insert a new pattern $a^j$ to the dictionary.
  Then, the total number $\utotal$ of failure links to update for all $a^j$'s
  is
  \[ \utotal = xk + x(k-1) + x(k-2) + \cdots + x = xk(k+1)/2.
  \]
  Let $\dadd$ be the total length of patterns to insert to the initial dictionary,
  and $d$ the total length of the patterns 
  after adding all patterns to the initial dictionary.
  Then $d = xk + \dadd = xk + k(k+1)/2$.
  Hence, the number of failure links to update
  for each character in the added patterns $a^j$'s
  is 
  \[ \frac{\utotal}{\dadd} < \frac{\utotal}{d} = \frac{xk(k+1)}{2xk + k(k+1)} = \frac{x(k+1)}{2x+k+1} = \Omega\Big(\frac{xk}{x+k}\Big),
  \]
  which becomes $\Omega(k)$ by choosing $x = \Omega(k)$.
  Hence, for each $1 \leq m \leq k$,
  when we add pattern $a^m$ of length $m$ to the dictionary,
  $\fnum = \Omega(km)$ failure links need to be updated.
\end{proof}

The arguments in the above proof consider the semi-dynamic case where
only insertion of new patterns in supported.
However, if we delete all patterns after they have been inserted,
then exactly the same number of failure links need to be updated.
Hence, the same matching upper and lower bounds hold also for the dynamic case
with both insertion and deletion of patterns.


	\section{Conclusions and Future Work}

We proposed two new algorithms for dynamic dictionary matching,
based on DAWGs and AC-automata.

The semi-dynamic version of our first method,
which uses the DAWG, supports
updates (insertions of patterns)
in $O(m \log \sigma)$ time 
and pattern matching in $O(n \log \sigma + \occ)$ time,
while the dynamic version supports
updates (insertions and deletions of patterns)
in $O(m (\log d / \log \log d + \log \sigma))$ time
and pattern matching in $O(n (\log d / \log \log d + \log \sigma) + \occ\log d / \log \log d)$ time.
Our second method supports updating the AC-automaton
in $O(m\log \sigma + \fnum + \onum)$ time with the additional DAWG update time,
and pattern matching in $O(n \log \sigma + \occ)$ time.
Here, $m$, $\sigma$, $n$, $\occ$, $d$, $\fnum$, and $\onum$
respectively denote the length of the pattern to insert/delete,
the alphabet size, the length of the text,
the number of occurrences of patterns in the text,
the total length of the patterns,
the number of AC-automaton states on which
the failure functions need to be updated,
and the number of AC-automaton states on which
the output functions need to be updated.
Since $\fnum$ and $\onum$ are the minimum costs to explicitly update
the AC-automaton, our second method is
faster than any existing dynamic dictionary matching
algorithms based on AC-automata~\cite{ishizaki2012incremental,meyer1985incremental,tsuda1995incremental}.

An intriguing open question is whether or not one can achieve
$O(m \log \sigma)$ update time and $O(n \log \sigma + \occ)$ pattern
matching time for dynamic dictionary matching allowing for
both insertions and deletions of patterns.

    \section*{Acknowledgments}

The research of Diptarama, Ryo Yoshinaka, and Ayumi Shinohara is supported by Tohoku University Division for Interdisciplinary Advance Research and Education,
JSPS KAKENHI Grant Numbers JP15H05706, \\ JP24106010, and ImPACT Program of Council for Science, Technology and Innovation (Cabinet Office, Government of Japan). The research of Shunsuke Inenaga is in part supported by JSPS KAKENHI Grant Number 17H01697.


    \bibliographystyle{abbrv}
	\bibliography{ref}
	
\end{document}